\documentclass[a4paper,fleqn]{cas-sc}
\usepackage[numbers]{natbib}
\usepackage{amsmath,amssymb}
\usepackage{amsthm,amsfonts}
\usepackage{amsxtra,fancybox,ascmac}
\usepackage{charter}
\usepackage{dsfont}
\usepackage{algorithm}
 \usepackage{algorithmic}
 \usepackage{arydshln}
\usepackage{subcaption}

\newcommand{\argmin}[1][]{\operatorname*{\mathrm{argmin}_\mathnormal{#1}}}

\newtheorem{asmp}{Assumption}
\newtheorem{pb}{Problem}
\newtheorem{lemma}{Lemma}

\newtheorem{theorem}{Theorem}

\newtheorem{remark}{Remark}
\def\tsc#1{\csdef{#1}{\textsc{\lowercase{#1}}\xspace}}
\tsc{WGM}
\tsc{QE}
\tsc{EP}
\tsc{PMS}
\tsc{BEC}
\tsc{DE}
\usepackage{float}
\begin{document}
\let\WriteBookmarks\relax
\def\floatpagepagefraction{1}
\def\textpagefraction{.001}

\shorttitle{Robust Distributed Sub-Optimal Coordination of Linear Agents with Uncertain Input Nonlinearities}

\shortauthors{T. Namba}

\title [mode = title]{Robust Distributed Sub-Optimal Coordination of Linear Agents with Uncertain Input Nonlinearities}                      
\tnotemark[1]

\tnotetext[1]{The author is supported by JSPS KAKENHI Grant-in-Aid for Research Activity Start-up 24K22960 and KAKENHI Grant-in-Aid for Early-Career Scientists 26K17405.}

\author[1]{Takumi Namba}[orcid=0009-0008-2906-2711]

\ead{t-namba@fc.ritsumei.ac.jp}


\credit{Conceptualization of this study, Methodology, Formal analysis, Writing - Original Draft, Funding acquisition}

\affiliation[1]{organization={Ritsumeikan University},
    addressline={1-1-1 Noji-Higashi}, 
    city={Kusatsu, Shiga},
    postcode={525-0058}, 
    country={Japan}}

\begin{abstract}
In this paper, we study robust distributed sub-optimal coordination of linear agents subject to input nonlinearities. Inspired by the robust agreement literature, we formulate a bounded distributed sub-optimal coordination problem, in which each agent converges to a neighborhood of the optimizer of a global optimization problem defined over a communication network. We propose a novel control protocol, and analyze convergence by employing a robust control approach, in which both the input nonlinearities and the gradients of the objective functions are treated in a unified manner via sector conditions. In particular, we derive sufficient conditions for the solvability of the considered problem and characterize them in terms of matrix inequalities. The effectiveness of the proposed method is demonstrated through a numerical simulation.
\end{abstract}




\maketitle
\section{Introduction}
\emph{Continuous-time} optimization algorithms have been extensively studied. Especially, in the last two decades, continuous-time \emph{distributed optimization over networks} has been developed using various control theoretical approaches \citep{Tao-Yang,Wang-CDC,Yu-IET,Deng-IET}. In traditional control frameworks, one solves the distributed optimization problem offline and obtains the optimizer, and then designs a local tracking controller to it, treating the two tasks separately. For systems that operate in real time, it is required to establish a distributed algorithm that asymptotically steers all the agents to the optimizer of the optimization problem over the network \emph{online}.

Recently, control-oriented distributed optimization over networked multi-agent systems has gained attention under names such as \emph{distributed optimal coordination}~\citep{Xian-TAC,Xian-SCL,An-TAC,Li-TAC,Yu-Automatica,Deng-TAC,Li-Auto,Huang-IET,Rahimi-IET,Wang-IET}, \emph{optimal consensus} \citep{Zhang-Cyb}, or \emph{distributed feedback optimization}~\citep{Terpin,Romano2,Romano-LMI}. Specifically, \emph{distributed optimal coordination} can be viewed as an \emph{agreement} (consensus/synchronization) problem in which every agent achieves the state/output agreement to the optimizer of a coupled optimization problem. 

One of the crucial requirements in the aforementioned framework is \emph{how to handle the agent dynamics and various physical characteristics embedded in them}. In practical scenarios, agents are operated online with an optimizing protocol, making hardware characteristics critical. Hardware characteristics range from input nonlinearities, input saturations, or parametric uncertainties. They are ubiquitous in various engineering applications like robotic networks, power systems, and so on. Nevertheless, this perspective has not been discussed in the distributed optimal coordination literature. 

\paragraph*{\textbf{Contribution:}}
In this paper, we study \emph{robust} distributed sub-optimal coordination of linear agents with input nonlinearities. To the best of the author's knowledge, this paper presents the first result on the robust (bounded) distributed sub-optimal coordination against the input nonlinearities. In the presence of the input nonlinearities, the exact convergence to the optimizer may be impossible. Motivated by this, we introduce a notion of \emph{robust distributed sub-optimal coordination}, in which each agent approaches and still remains in a certain region around the optimizer. We propose a novel control protocol that solves the aforementioned problem, and analyze the convergence via the input-to-state stability argument. Moreover, we derive a gain condition that provides a sufficient condition for the convergence, which is expressed in terms of matrix inequalities.

The contributions of this work are summarized as follows: 
\begin{itemize}
    \item{We formulate a novel robust practical (bounded) distributed sub-optimal coordination problem taking into account the input nonlinearities embedded in the agent dynamics.}
    \item{We propose a novel distributed protocol for distributed sub-optimal coordination of linear agents with input nonlinearities.}
    \item{We derive a sufficient condition for the solvability and the agreement tolerance, which is expressed in terms of linear matrix inequalities. Based on an input-to-state stability (ISS) analysis, we prove that the proposed approach guarantees that the robust distributed sub-optimal coordination is achieved against the input nonlinearities.}
   \item{We demonstrate the effectiveness of the proposed method through a numerical example.}
\end{itemize}

The above contributions extend the applicability of robust distributed optimal coordination approaches to more general and practical agents.
\paragraph*{\textbf{Comparison with Existing Results:}}
Robust agreement problems have been studied in \citet{Trentelman,Nguyen,Takaba-JCMSI,Takaba-ICCAS}.
For a robust agreement problem of physically coupled agents, \citet{Namba} proposed a scaled small-gain approach for less conservative robustly agreeing controller synthesis. Nevertheless, optimization-aware robust agreement remains unexplored.

\citet{Yu-Chen} and the very recent publication by  \citet{11386952} dealt with the robust distributed sub-optimal coordination of the linear agents and nonlinear agents, respectively, and studied  robustness of the algorithms against the external disturbances via the ISS-based approaches. Nonetheless, our analysis is completely different from theirs in terms of the disturbances considered in \citep{Yu-Chen,11386952}. \citet{Yu-Chen,11386952} considered external disturbances that do not depend on the agent state, and assumed them to be norm bounded. Therefore, their framework cannot capture the input-dependent nonlinearities as considered in this paper.  

As a preliminary result related to this paper, \citet{CAO} studied a local agreement-based distributed optimal coordination of the first-order linear agents with input saturation. However, the result was limited to a specific first-order setting, and the major contributions of the present paper were not addressed in \citep{CAO}.

\paragraph{\textbf{Notations:}} For a vector $x$, the $j$-th component of $x$ is denoted by $[x]_j$. All-one vector is denoted by $\mathds{1}:=[1,1,\ldots,1]^{\sf T}$. For a symmetric matrix $M$, $\underline{\lambda}(M)$ and $\overline{\lambda}(M)$ denote the smallest and largest eigenvalues.

\section{Problem Formulation}
\label{sec:2}
\subsection{Multi-Agent Systems}
In this paper, we consider a multi-agent system consisting of $N$ linear agents, and an optimization problem defined over them.
Each agent has time-varying input nonlinearity and its dynamics is expressed by
\begin{align}
\label{eq:dynamics}
    \dot{x}_i=Ax_i+B_0\phi_i(u_i,t), \ i=1,\ldots,N,
\end{align}
where $x_i(t)\in\mathds{R}^{n}$ is the state, $u_i(t)\in\mathds{R}^m$ is the input, and $A\in\mathds{R}^{n\times n}$, and $B_0\in\mathds{R}^{n\times m}$ are the constant matrices, and $\phi_i(u_i,t):\mathds{R}^{m}\times \mathds{R}_{\geq0}\rightarrow \mathds{R}^m$ is the time-varying input nonlinearity. The input nonlinearity $\phi_i(u_i,t)$ is ubiquitous in various practical applications, to model uncertainties and heterogeneity in robust control theory. Throughout this paper, we sometimes use the notation $\phi_i(u_i)$, dropping $t$ for simplicity. 
\begin{asmp}
\label{asmp:stabilizability}
The pair $(A,B_0)$ is stabilizable, and the matrix $B_0$ is full row rank.
\end{asmp}
Note that the distributed optimal coordination is a kind of reference tracking problem to the optimizer with respect to the full state. The above rank condition of $B_0$ is required to achieve such full-state tracking. 
\begin{asmp}
\label{asmp:input-nonlinearity}
The input nonlinearity $\phi_i(u_i,t)$ globally satisfies the component-wise incremental sector condition, i.e., 
\begin{align}
\alpha\leq\frac{[\phi_i(u_i,t)-\phi_i(u_i',t)]_j}{[u_i-u_i']_j}\leq\beta  \ \ \ \ \forall j =1,\ldots,m,\ \forall t\geq0, \ \forall u_i,u_i'\in\mathds{R}^m, \ \forall i=1,\ldots,N.
\end{align}
Moreover, $[\phi_i(0,t)]_j=0 \ \  \forall i=1,\ldots,N,~\forall j=1,\ldots,m,~\forall t\geq0$.
\end{asmp}
The incremental sector-bounded property is standard in the literature, such as \cite{Zhang,Han-AJC}. Also note that Assumption~\ref{asmp:input-nonlinearity} implies that $\phi_i(u,t)$ is sector-bounded with $[\alpha,\beta]$ in the usual sense. To facilitate the subsequent analysis, we decompose the input nonlinearity into a linear component and a residual nonlinear term. To this aim, we define the function $\phi_i'(u_i,t):\mathds{R}^m\times \mathds{R}_{\geq 0}\rightarrow \mathds{R}^m$ as 
\begin{align}
    \label{eq:def-phidash}
    \phi_i'(u_i,t):=u_i-\frac{1}{\beta}\phi_i(u_i,t)\color{black}.
\end{align}
Then, the dynamics \eqref{eq:dynamics} is rewritten as 
\begin{align}
    \dot{x}_i&=Ax_i+Bu_i-\color{black}B\phi_i'(u_i,t), \ i=1,\ldots,N
\end{align}
with $B=\beta B_0$. One can verify that the function $\phi_i'(u_i,t)$ satisfies the sector bound condition with the sector $[0,\gamma]$ with $\gamma=(\beta-\alpha)/\beta$. (Refer to \cite{Takaba-ICCAS}). This immediately implies the sector conditions
\begin{align}
    \phi'(u)^{\sf T}(\phi'(u)-\gamma u)\leq 0 \ \ \ \forall t\geq 0, \forall u\in\mathds{R}^{mN}
\end{align}
with 
\[
u:=
\begin{bmatrix}
    u_1\\
    u_2\\
    \vdots\\
    u_N
\end{bmatrix}, \ 
\phi'(u,t):=
\begin{bmatrix}
    \phi_1'(u_1,t)\\
    \phi_2'(u_2,t)\\
    \vdots\\
    \phi_N'(u_N,t)
\end{bmatrix}.
\]
Moreover, $\phi'(u,t)$ satisfies the \emph{incremental} sector condition as follows:
\begin{align}
\label{eq:incr-sect-bound-input}
\{\phi'(u',t)-\phi'(u,t)\}^{\sf T}\{(\phi'(u',t)-\phi'(u,t))-\gamma (u'-u)\}\leq 0,
\end{align}
for any inputs $u,u'\in\mathds{R}^{mN}$, and $t\geq 0$. 

\subsection{Local Objective Function}
As described in Introduction, we consider the situation where the agents cooperatively find the optimizer of a certain convex optimization problem, and each agent has its own local objective function $f_i(x_i):\mathds{R}^n\rightarrow\mathds{R}$.
\begin{asmp}
\label{asmp:obj-privacy}
The local objective function $f_i(x_i)$ is only available to the agent $i$. Further, the agent $i$ cannot access the information of the other agents' objective functions $f_j(x_j),~j\neq i$.
\end{asmp}
\begin{asmp}
\label{asmp:objf}
The objective function $f_i(x_i)$ is continuously differentiable, and strongly convex in $x_i$. Moreover, the gradient $\nabla f_i(x_i)$ is globally Lipschitz. 
\end{asmp}
Under Assumption~\ref{asmp:objf}, there exists a constant $\mu$ satisfying \[(\nabla f_i(x_i')-\nabla f_i(x_i))^{\sf T}(x_i'-x_i)\geq \mu\|x_i'-x_i\|^2\] for all $x_i',x_i\in\mathds{R}^{n}$, and there exists a constant $\ell$ satisfying 
\begin{align}
\label{eq:ell-Lip}
    (\nabla f_i(x_i')-\nabla f_i(x_i))^{\sf T}(x_i'-x_i)\leq \ell(x_i'-x_i)^{\sf T}(x_i'-x_i)
    \end{align}
for all $x_i',x_i\in\mathds{R}^{n}$ (Refer to e.g. \citet{Bertsekas}).

We reformulate the gradient term of the objective function so that it can be treated within the same robust-theoretic framework as the input nonlinearity. We thus define
\begin{align*}
\psi_i(x_i):=x_i-\frac{1}{\ell}\nabla_{i}f_i(x_i),~~
\psi(x):=x-\frac{1}{\ell}\nabla f(x).
\end{align*}
Notice that ${\psi}(x)$ satisfies the following incremental sector conditions with the sector $[0,\mu']$ as   
\begin{equation}
\label{eq:Psi-sect}
({\psi}(x')-{\psi}(x))^{\sf T}[({\psi}(x')-{\psi} (x))-\mu' (x'-x)]\leq 0
\end{equation}
with $\mu':=(\ell-\mu)/\ell$, for all $x',x\in\mathds{R}^{nN}$.

\subsection{Robust Distributed Sub-optimal Coordination}
Inspired by robust agreement literature, we introduce a robust distributed sub-optimal coordination problem as a bounded (practical) agreement problem to the optimizer of the networked optimization problem. Under the input nonlinearity constraints, it may be difficult to achieve \emph{exact} convergence to the optimizer.
\begin{pb}
\label{prob:1}
For any initial state $x_1(0),\ldots,x_N(0)$, design a distributed dynamical protocol such that
\begin{align*}
     &\limsup_{t\rightarrow \infty}\|x(t)-\mathds{1}_N\otimes x^{\star}\|\leq\epsilon \ \ \ \ \ ~\forall x_i(0)\in\mathds{R}^n,
    \\
    &\mathrm{where~~}x^{\star} :=\displaystyle \argmin_z f(z)=\sum_{i=1}^{N}f_i(z)~\cdots(\mathbf{OP}),
\end{align*}
holds for some constant $\epsilon$, which means that the agents' states approach and remain a bounded region around the optimizer of $(\mathbf{OP})$ as time goes to infinity.  
\end{pb}
It should be noted that, under Assumption \ref{asmp:objf}, the optimization problem $(\mathbf{OP})$ has a unique optimizer $x^\star$. Also, the existence of the constant $\epsilon$ can be viewed as a sub-optimality certificate.
\section{Local Protocol with Communication}
\subsection{Information Exchange Model}
In this subsection, we briefly review algebraic graph theory for the network model of the multi-agent system, due to \citet{Bullo2018,Godsil}.

Information exchange among agents is expressed as a weighted undirected graph $\mathcal{G}:=(\mathcal{V},\mathcal{E},\mathcal{W})$. The set $\mathcal{V}:=\{1,2,\ldots,N\}$ represents the set of vertices (agents), $\mathcal{E}\subseteq\mathcal{V}\times \mathcal{V}$ is an edge set (communication links), and $\mathcal{W}$ is a set of edge weights. If an edge $(j,i)\in\mathcal{E}$ exists, the agent $i$ can receive the information from the agent $j$. The set of neighboring agents of $i$ is defined by $\mathcal{N}_i:=\{j\in\mathcal{V}~|~(j,i)\in\mathcal{E},~j\neq i\}$. We introduce an adjacency matrix $\mathsf{A}:=({a}_{ij})\in\mathds{R}^{N\times N}$ and a degree matrix $\mathsf{D}:=({d}_{ij})\in\mathds{R}^{N\times N}$ by ${a}_{ij}=
        \omega_{ji}$ $ \mbox{if }(j,i)\in\mathcal{E}$, and otherwise ${a}_{ij}= 
        0$, and $ 
    {d}_{ij}=\sum_{j\in\mathcal{N}_i}{a}_{ij} $ if $ i=j$, otherwise ${d}_{ij}=0$, where $\omega_{ji}>0$ denotes the weight of an edge $(j,i)$. We consider the undirected graph case, so ${a}_{ij}={a}_{ji}$ for $(j,i)\in\mathcal{E}$ holds true. We further define the weighted graph Laplacian $L\in\mathds{R}^{N\times N}$ by $L:=\mathsf{D}-\mathsf{A}$.
\begin{asmp}
The communication graph $\mathcal{G}$ is connected, and its topology is time-invariant. 
\end{asmp}
Under the above assumption, $L$ is a positive semi-definite symmetric matrix with at least one zero eigenvalue, and $L{\mathds{1}}_N=0$ holds true with ${\mathds{1}}_{N}:=[1,1,\ldots,1]^{\sf T}\in\mathds{R}^N$. Hereafter, we denote the eigenvalues of $\mathcal{G}$ by $\lambda_{i}$, and arrange them in the ascending order as $0=\lambda_{1}\leq \lambda_{2}\leq \lambda_{3}\leq \cdots\leq \lambda_{N}$. Further, we introduce an orthogonal matrix $U\in\mathds{R}^{N\times N}$ in the form $
U:=\begin{bmatrix}\frac{\mathds{1}_N}{\sqrt{N}} & X_{1} \end{bmatrix}
,~X_{1}\in\mathds{R}^{N\times (N-1)}$ that diagonalizes $L$ as \[ 
U^{\sf T}LU=\Lambda:=
\mathrm{diag}\{\lambda_{1},\lambda_{2},\cdots,\lambda_{N}\}.
\]
Note that $\|U\|=1$ and $\|X_{1}\|=1$ hold from the orthogonality of $U$. 
\subsection{Distributed Dynamic Protocol}
In this paper, we propose a novel dynamic protocol that solves the distributed sub-optimal coordination. It should be emphasized that, different from the standard agreement problem, the proposed controller has to seek the optimizer through local communication, while the local objective function itself is not revealed to the other agents. Thus, the problem addressed here is much more complex and more difficult than the standard robust agreement problems.

To achieve the aforementioned objectives, we design the local protocol in the form:
\begin{subequations}
\label{eq:dyncamic-protocol}
\begin{align}
u_i&=K_1x_i+K_2v_i+K_3\zeta_i+K_4\eta_i,
\label{eq:dyncamic-protocol-1}
\\
\dot{v}_i&=\sum_{j\in\mathcal{N}_i}a_{ij}(x_j-x_i),
\label{eq:dyncamic-protocol-2}\\
\dot{\zeta}_i&=\sum_{j\in\mathcal{N}_i}a_{ij}(x_j-x_i)+v_i-\nabla f_i(x_i),
\label{eq:dyncamic-protocol-3}\\
\dot{\eta}_i&=x_i-\zeta_i
\label{eq:dyncamic-protocol-4}
\end{align}
\end{subequations}
with $v_i(t)\in\mathds{R}^{n}$, $\zeta_i(t)\in\mathds{R}^{n}$, $\eta_i(t)\in\mathds{R}^n$, where the gains $K_1,\ldots,K_4\in\mathds{R}^{m\times n}$ are to be designed. As described later in Section~\ref{sec:4}, this protocol generates a surrogate for the optimizer, and is designed so that the state tracks this surrogate in a distributed fashion.

\begin{asmp}
\label{asmp:initial-setting}
The initial value of the controller state $v_i(0)$ is chosen such that $\sum_{i=1}^{N} v_i(0)=0$. 
\end{asmp}
As a typical choice, one can take $v_i(0)=0,\ i=1,\ldots,N$. This assumption is required to establish the result in Lemma~\ref{lem:1}.
\begin{remark}
Parts of the dynamic protocol \eqref{eq:dyncamic-protocol-2}-\eqref{eq:dyncamic-protocol-3} are typical choices in the distributed optimal coordination literature \citep{Li-TAC,Zhang-Cyb}.
Nevertheless, the whole controller structure \eqref{eq:dyncamic-protocol} is well-tailored for ease of robust-theoretic analysis as in this paper, motivated by the LMI-based optimal steady-state control in Nash seeking problems \citep{Romano2,Romano-LMI}. Thanks to a novel controller structure presented here, we can naturally apply robust control techniques against the sector-bounded nonlinearities. 
\end{remark}
\section{Optimality and Convergence Analysis}
\label{sec:4}
\subsection{Optimality Analysis}
Firstly, notice that the closed-loop multi-agent systems \eqref{eq:dynamics} under the control protocol \eqref{eq:dyncamic-protocol} is expressed as 
\begin{subequations}
\label{eq:collective-1}
\begin{align}
\dot{x}&=(I\otimes A)x+(I\otimes B)u-(I\otimes B)w,\label{eq:dynamics-with-w1}\\
\dot{v}&=-(L\otimes I)x,
\label{eq:dynamics-with-v}
\\
\dot{\zeta}&=-(L\otimes I)x+v-\nabla f(x),\\
\dot{\eta}&=x-\zeta,\\
w&=\phi'(u,t)
\end{align}
\end{subequations}
with the collective notation
\begin{align}
    x:=
    \begin{bmatrix}
        x_1\\
        x_2\\
        \vdots\\
        x_N
    \end{bmatrix}, \ 
        u:=
    \begin{bmatrix}
        u_1\\
        u_2\\
        \vdots\\
        u_N
    \end{bmatrix}, \ 
        w:=
    \begin{bmatrix}
        w_1\\
        w_2\\
        \vdots\\
        w_N
    \end{bmatrix}, \ 
        \zeta:=
    \begin{bmatrix}
        \zeta_1\\
        \zeta_2\\
        \vdots\\
        \zeta_N
    \end{bmatrix},\ 
        \eta:=
    \begin{bmatrix}
        \eta_1\\
        \eta_2\\
        \vdots\\
        \eta_N
    \end{bmatrix}
    .
\end{align}
The next lemma establishes the relationship between the optimizer and the constant solution vector to a certain simultaneous equation.
\begin{lemma}
\label{lem:1}
Assume that there exist the constant vectors $\overline{x}\in\mathds{R}^{nN},\overline{v}\in\mathds{R}^{nN},\overline{\zeta}\in\mathds{R}^{nN},\overline{\eta}\in\mathds{R}^{nN}$ and $\overline{u}\in\mathds{R}^{mN}$ satisfying:
\begin{subequations}
\label{eq:steady-state1}
\begin{align}
0&=(I\otimes A)\overline{x}+(I\otimes B)\overline{u},\label{eq:ss1}\\
0&=-\color{black}(L\otimes I)\overline{x},\label{eq:ss2}\\
0&=-\color{black}(L\otimes I)\overline{x}+\overline{v}-\nabla f(\overline{x}),\label{eq:ss3} \\
0&=\overline{x}-\overline{\zeta},\label{eq:ss4}\\
\overline{u}&=(I\otimes K_1)\overline{x}+(I\otimes K_2)\overline{v}+(I\otimes K_3)\overline{\zeta}+(I\otimes K_4)\overline{\eta}.
\end{align}
\end{subequations}
Then, the constant vector $\overline{x}$ is given by $\overline{x}=\mathds{1}_N\otimes x^\star$ uniquely, where $x^{\star}$ coincides with the (unique) optimizer of $(\mathbf{OP})$.
\end{lemma}
\begin{proof}
We prove this lemma by following a similar line to \citep{Li-TAC}. From \eqref{eq:ss2}, $\overline{x}\in\mathrm{ker}(L)$ which implies $\overline{x}\in\mathrm{im}\{\mathds{1}\}$. Thus $\overline{x}$ can be expressed as $\overline{x}=\mathds{1}_N\otimes x^\#$ with a certain $x^\#$. Also, multiplying $(\mathds{1}^{\sf T}\otimes I)$ to \eqref{eq:ss3} yields 
\begin{align}
\label{eq:first-order-optimality}
    0=(\mathds{1}^{\sf T}\otimes I)\nabla f(\overline{x})=\sum_{i=1}^{N}\nabla_i f_i(x^\#)=0,
\end{align}
where we used $\sum_{i=1}^n v_i(t)=(\mathds{1}\otimes I)v(t)\equiv0~\forall t\geq 0$ under Assumption~\ref{asmp:initial-setting} and \eqref{eq:dynamics-with-v}. Further, under Assumption~\ref{asmp:objf}, since $\sum_{i=1}^Nf(z)$ is strongly convex with respect to $z:=x_1=x_2=\cdots=x_N$, and \eqref{eq:first-order-optimality} is a necessary and sufficient condition that $x^\#$ is the optimizer of $(\mathbf{OP})$. 
\end{proof}
\begin{remark}
The algebraic condition \eqref{eq:steady-state1} characterizes a particular equilibrium (steady-state) candidate of \eqref{eq:collective-1}, and it corresponds to an actual equilibrium only when $\phi'(\overline{u},t)\equiv 0~\forall t\geq 0$ is satisfied, including the case without input nonlinearities. In the presence of possibly time-varying input nonlinearities, this condition is generally not consistent with the original dynamics \eqref{eq:dynamics-with-w1}, which means that the equilibrium point does not exist in general, since the nonlinear term $w(\overline{u},t)$ does not vanish. Hence, the exact convergence to the optimizer is not achievable, which motivates the robust practical (bounded) distributed optimal coordination formulation in Problem~\ref{prob:1}.
\end{remark}
\subsection{Convergence Analysis}
In this subsection, we analyze the convergence properties of the closed-loop system and show that the proposed controller solves the robust distributed optimal coordination problem (Problem \ref{prob:1}). Following the discussion in Lemma~\ref{lem:1}, the constant vector $(\overline{x},\overline{v},\overline{\zeta},\overline{\eta})$ serves as the reference point for the subsequent error analysis. Particularly, our convergence analysis is based on a shifted system around the reference point, a coordinate transformation associated with the graph Laplacian, and an ISS argument (refer to e.g. \citet{Khalil2002}). 

Let us define the shifted variables $\tilde{x}:=x-\overline{x},\ \tilde{v}:=v-\overline{v}, \ \tilde{\zeta}:=\zeta-\overline{\zeta}, \ \tilde{\eta}:=\eta-\overline{\eta}$, and $\tilde{u}:=u-\overline{u}$. We then introduce the shifted system: 
\begin{subequations}
\label{eq:collective-error}
\begin{align}
\dot{\tilde{x}}&=(I\otimes A)\tilde{x}+(I\otimes B)\tilde{u}-(I\otimes B)\tilde{w}+(I\otimes B){w}^\star\\
\dot{\tilde{v}}&=-\color{black}(L\otimes I)\tilde{x}\\
\dot{\tilde{\zeta}}&=-\color{black}(L\otimes I)\tilde{x}+\tilde{v}-\ell \tilde{x}+\ell \tilde{g}\\
\dot{\tilde{\eta}}&=\tilde{x}-\tilde{\zeta},\\
\tilde{u}&=(I\otimes K_1)\tilde{x}+(I\otimes {K}_2)\tilde{v}+(I\otimes K_3)\tilde{\zeta}+(I\otimes K_4)\tilde{\eta}
\end{align}
\end{subequations}
where 
\begin{subequations}
\begin{align*}
    \tilde{w}(t)&:=\phi'(u,t)-\phi'(\overline{u},t),\\
    \tilde{g}(t)&:=\psi(x)-\psi(\overline{x}),\\
    w^\star(t)&:=\phi'(\overline{u},t).
\end{align*}
\end{subequations}
It should be noted that the input nonlinearity $\phi'(u,t)$ is possibly time-varying, and thus $w^\star(t)$ is a time-varying signal, while $\overline{u}$ is a constant vector.

Further, to exploit the block-diagonal structure, 
applying the coordinate transformation with $U\otimes I$ to \eqref{eq:collective-error} yields
\begin{subequations}
\begin{align}
\dot{\boldsymbol{x}}&=(I\otimes A)\boldsymbol{x}+(I\otimes B)\boldsymbol{u}-(I\otimes B)\boldsymbol{w}+(I\otimes B)\boldsymbol{w}^\star\\
\dot{\boldsymbol{v}}&=-\color{black}(\Lambda\otimes I)\boldsymbol{x}\\
\dot{\boldsymbol{\zeta}}&=-\color{black}(\Lambda\otimes I)\boldsymbol{x}+\boldsymbol{v}-\ell \boldsymbol{x}+\ell \boldsymbol{g}\\
\dot{\boldsymbol{\eta}}&=\boldsymbol{x}-\boldsymbol{\zeta},\\
\boldsymbol{u}&=(I\otimes K_1)\boldsymbol{x}+(I\otimes K_2)\boldsymbol{v}+(I\otimes K_3)\boldsymbol{\zeta}+(I\otimes K_4)\boldsymbol{\eta},
\\
\boldsymbol{w}&=(U^{\sf T}\otimes I)\{\phi'(u,t)-\phi'(\overline{u},t)\},
\label{eq:w-def-incr}
\\
\boldsymbol{g}&=(U^{\sf T}\otimes I)\{{\psi}(x,t)-\psi(\overline{x},t)\}
\label{eq:g-def-incr}
\end{align}
\end{subequations}
with $\boldsymbol{x}:=(U^{\sf T}\otimes I)\tilde{x}, \ \boldsymbol{v}:=(U^{\sf T}\otimes I)\tilde{v}, 
\ \boldsymbol{\zeta}:=(U^{\sf T}\otimes I)\tilde{\zeta}, 
\ \boldsymbol{\eta}:=(U^{\sf T}\otimes I)\tilde{\eta}, 
\ \boldsymbol{u}:=(U^{\sf T}\otimes I)\tilde{u}$, $\boldsymbol{w}^\star:=(U^{\mathsf{T}}\otimes I)w^\star$. Recall that $\sum_{i=1}^Nv_i(0)=0$ from Assumption~\ref{asmp:initial-setting}. This can be rewritten as $(\mathds{1}^{\sf T}\otimes I)\boldsymbol{v}=0$, namely, $\boldsymbol{v}_1(t)\equiv0\ \forall t\geq 0$. We define a new state vector $\boldsymbol{v}_{2:N}:=[\boldsymbol{v}_2^{\sf T},\ldots,\boldsymbol{v}_N^{\sf T}]^{\sf T}$ by removing $\boldsymbol{v}_1$ from $\boldsymbol{v}$, and we rewrite the collective representation as
\begin{subequations}
\label{eq:collective-dynamics1}
\begin{align}
\dot{\xi}_1&=
\begin{bmatrix}
A&0&0 \\
-\ell I & 0 & 0 \\
I & -I & 0
\end{bmatrix}
\xi_1+
\begin{bmatrix}
    B\\ 0 \\0
\end{bmatrix}
\boldsymbol{u}_1-
\begin{bmatrix}
B & 0\\
0 & -\ell I \\
0 & 0
\end{bmatrix}
\begin{bmatrix}
\boldsymbol{w}_1\\
    \boldsymbol{g}_1
\end{bmatrix}+
\begin{bmatrix}
B \\ 0 \\ 0
\end{bmatrix}
\boldsymbol{w}_1^\star
, 
\\
\dot{\xi}_i&=
\begin{bmatrix}
    A& 0 & 0 & 0\\
    -\lambda_i I & 0 & 0& 0\\
    -(\lambda_i+\ell) I & I & 0  & 0\\
    I & 0 & -I & 0
\end{bmatrix}
\xi_i+
\begin{bmatrix}
    B \\ 0 \\ 0\\0 
\end{bmatrix}
\boldsymbol{u}_i-
\begin{bmatrix}
B & 0 \\
0 & 0 \\
0 & -\ell I \\0 & 0
\end{bmatrix}
\begin{bmatrix}
\boldsymbol{w}_i\\
\boldsymbol{g}_i
\end{bmatrix}+
\begin{bmatrix}
    B \\ 0 \\ 0 \\ 0
\end{bmatrix}
\boldsymbol{w}_i^\star
\\
\boldsymbol{w}&=(U^{\sf T}\otimes I)\{\phi'(u,t)-\phi'(\overline{u},t)\},\\
\boldsymbol{g}&=(U^{\sf T}\otimes I)\{{\psi}(x,t)-\psi(\overline{x},t)\}
\end{align}
\end{subequations}
where the state vectors are rearranged as 
\begin{align}
    \xi_1:=\begin{bmatrix}
        \boldsymbol{x}_1\\
        \boldsymbol{\zeta}_1\\
        \boldsymbol{\eta}_1
    \end{bmatrix},
    \
    \xi_i:=\begin{bmatrix}
        \boldsymbol{x}_i\\
        \boldsymbol{v}_i\\
        \boldsymbol{\zeta}_i\\
        \boldsymbol{\eta}_i
    \end{bmatrix},\ i=2,\ldots,N.
\end{align}

The closed-loop multi-agent system can also be expressed as 
\begin{subequations}
\label{eq:collective-dynamics2}
\begin{align}
\dot{\xi}_1&=(\mathcal{A}_1+\mathcal{B}_1 \check{K})\xi_1-[\mathcal{B}_1~\mathcal{L}_1]\begin{bmatrix}
    \boldsymbol{w}_1\\
    \boldsymbol{g}_1
\end{bmatrix}
+\mathcal{B}_1\boldsymbol{w}_1^{\star}\\
\dot{\xi}_i&=:(\mathcal{A}_i+\mathcal{B}' K)\xi_i-[\mathcal{B}'~\mathcal{L}']\begin{bmatrix}
    \boldsymbol{w}_i\\
    \boldsymbol{g}_i
\end{bmatrix}
+\color{black}\mathcal{B}' \boldsymbol{w}_i^\star,\notag\quad\quad\quad i=2,\ldots,N, \\
\boldsymbol{w}&=(U^{\sf T}\otimes I)\{\phi'(u,t)-\phi'(\overline{u},t)\},\\
\boldsymbol{g}&=(U^{\sf T}\otimes I)\{{\psi}(x,t)-\psi(\overline{x},t)\}
\end{align}
\end{subequations}
with $K:=[K_1,K_2,K_3,K_4]$ and $\check{K}:=[K_1,K_3,K_4]$, where the coefficient matrices are defined by 
\begin{align*}
\mathcal{A}_1&:=
\begin{bmatrix}
    A & 0 & 0 \\
    -\ell I & 0 &0\\
    I & -I & 0
\end{bmatrix}, 
\ 
\mathcal{A}_i:=
\begin{bmatrix}
A & 0 & 0 & 0\\
-\lambda_i I & 0 & 0 &0 \\
-(\lambda_i+\ell)I & I & 0 & 0 \\
I & 0 & -I & 0
\end{bmatrix},\\
\mathcal{B}_1&:=\begin{bmatrix}
    B \\ 0 \\0 
\end{bmatrix}, 
\ 
\mathcal{B}':=\begin{bmatrix}
    B \\ 0 \\0  \\ 0
\end{bmatrix}, \ 
\mathcal{L}_1:=\begin{bmatrix}
    0 \\ -\ell I \\ 0
\end{bmatrix}, 
\ 
\mathcal{L}':=\begin{bmatrix}
    0 \\ 0 \\-\ell I \\ 0
\end{bmatrix}, \ i=2,\ldots,N.
\end{align*}

One can observe that, in the shifted system \eqref{eq:collective-dynamics2}, $\boldsymbol{w}(t)$ and $\boldsymbol{g}(t)$ are viewed as the \emph{inner-loop} feedback inputs passed through the sector-bounded nonlinearity. On the other hand, the term $\boldsymbol{w}^\star(t)$ emerges as a \emph{fictitious external} input, which may be time-varying. Accordingly, we will perform the ISS analysis with respect to $\boldsymbol{w}^{\star}(t)$. To be more precise, the \emph{solvability} of Problem~\ref{prob:1} is reduced to the input-to-state stability of \eqref{eq:collective-dynamics2} with respect to $\boldsymbol{w}^\star(t)$, which is a bounded fictitious external input (as we observe later). We state the main result of this paper as follows.
\begin{theorem}
\label{thm:1}
Under Assumptions~\ref{asmp:stabilizability}-\ref{asmp:initial-setting}, for a given feedback gain $K$, assume that there exists a positive definite solution $P=P^{\sf T}\in\mathds{R}^{4n\times 4n}$ and a positive definite solution $\check{P}=\check{P}^{\sf T}\in\mathds{R}^{3n\times 3n}$ to the following matrix inequalities:
\begin{align}
\label{eq:LMI1}
&\begin{bmatrix}
    (\mathcal{A}_{i}+\mathcal{B}'K)^{\sf T}P+P(\mathcal{A}_i+\mathcal{B}'K) & 
    \mathcal{K}'^{\sf T}-P
    \begin{bmatrix}\mathcal{B}' & \mathcal{L}'\end{bmatrix}\\
    \mathcal{K}'-
    \begin{bmatrix}
    \mathcal{B}'^{\sf T}\\
    \mathcal{L}'^{\sf T}
    \end{bmatrix}P& -2I
\end{bmatrix}
\prec 0, \ \ \ i=2,\ldots,N,
\end{align}
and 
\begin{align}
    \label{eq:LMI2}
&\begin{bmatrix}
 (\mathcal{A}_1+\mathcal{B}_1\check{K})^{\sf T}\check{P}+\check{P}(\mathcal{A}_1+\mathcal{B}_1\check{K}) & \check{\mathcal{K}}^{\sf T}-\check{P}\begin{bmatrix}
     \mathcal{B}_1 & \mathcal{L}_1
 \end{bmatrix}\\
 \check{\mathcal{K}}-
 \begin{bmatrix}
     \mathcal{B}_1^{\sf T}\\
     \mathcal{L}_1^{\sf T}
 \end{bmatrix}
 \check{P} & -2I
\end{bmatrix}
\prec 0,
\end{align}
where 
\begin{align*}
\mathcal{K}'&:=\begin{bmatrix}
    \gamma K\\  
    \mu' J
\end{bmatrix}, \ 
\check{\mathcal{K}}:=\begin{bmatrix}
    \gamma\check{K} \\ 
    \mu'\check{J}
\end{bmatrix},\ 
J:=\begin{bmatrix}
    I_n & 0_n & 0_n &0_n
\end{bmatrix}, \ 
\check{J}:=\begin{bmatrix}
    I_n & 0_n & 0_n
\end{bmatrix}.
\end{align*}
Then, the multi-agent system \eqref{eq:collective-dynamics1} solves the robust practical distributed optimal coordination problem (Problem~\ref{prob:1}). 
\end{theorem}

\begin{proof}
Let us introduce a candidate Lyapunov function $V(\xi)$ by
\begin{align}
    V(\xi)&:=\xi^{\sf T}\mathcal{P}\xi,
\end{align}
where $\mathcal{P}:=\mathrm{diag}\{\check{P},I_{N-1}\otimes P\}\succ 0$ with the solutions $\check{P}$ and $P$ to the LMIs \eqref{eq:LMI1}-\eqref{eq:LMI2}. 
Then, the derivative of $V(\xi)$ along the trajectory is given by 
\begin{align}
\dot{V}(\xi)=&2\sum_{i=2}^N \xi^{\sf T}_iP\left\{
(\mathcal{A}_i+\mathcal{B}'K)\xi_i-\begin{bmatrix}
    \mathcal{B}' & \mathcal{L}'
\end{bmatrix}\boldsymbol{\eta}_i+\mathcal{B}'\boldsymbol{w}_i^{\star}
\right\} \notag \\ 
&+2\xi_1^{\sf T}\check{P}\{(\mathcal{A}_1+\mathcal{B}_1\check{K})\xi_1-\begin{bmatrix}
    \mathcal{B}_1 & \mathcal{L}_1
\end{bmatrix}
\boldsymbol{\eta}_1+\mathcal{B}_1\boldsymbol{w}_1^\star
\}
    \label{eq:Vdot}
\end{align}
where $\boldsymbol{\eta}_{i}:=[\boldsymbol{w}_i^{\sf T},\boldsymbol{g}_i^{\sf T}]^{\sf T}$. Pre-multiplying $[\xi_i^{\sf T},\boldsymbol{\eta}_i^{\sf T}]$ and post-multiplying $[\xi_i^{\sf T},\boldsymbol{\eta}_i^{\sf T}]^{\sf T}$ to the left-hand side of \eqref{eq:LMI1} yield
\begin{align}
\label{eq:der1}
&\xi_i^{\sf T}\{(\mathcal{A}_i+\mathcal{B}'K)^{\sf T}P+P(\mathcal{A}_i+\mathcal{B}'K)\}\xi_i+\xi_i^{\sf T}(\mathcal{K}'^{\sf T}-P\begin{bmatrix}
    \mathcal{B}' & \mathcal{L}'
\end{bmatrix})\boldsymbol{\eta}_i+\boldsymbol{\eta}_i^{\sf T}\left(
\mathcal{K}'-\begin{bmatrix}
    \mathcal{B}'^{\sf T}\\
    \mathcal{L}'^{\sf T}
\end{bmatrix}
\right)\xi_i-2\boldsymbol{\eta}_i^{\sf T}\boldsymbol{\eta}_i\leq 0
\end{align}
for $i=2,\ldots,N$. Also from \eqref{eq:LMI2}, 
\begin{align}
\label{eq:der2}
\xi_1^{\sf T}\left\{(\mathcal{A}_1+\mathcal{B}_1\check{K})^{\sf T}\check{P}+\check{P}(\mathcal{A}_1+\mathcal{B}_1\check{K})\right\}+\xi_1^{\sf T}(\check{\mathcal{K}}^{\sf T}-\check{P}\begin{bmatrix}
\mathcal{B}_1 & \mathcal{L}_1
\end{bmatrix})\eta_1+\eta_1^{\sf T}
\left(\check{\mathcal{K}}-
\begin{bmatrix}
    \mathcal{B}_1^{\sf T}\\
    \mathcal{L}_1^{\sf T}
\end{bmatrix}
\check{P}
\right)\xi_1-2\boldsymbol{\eta}_1^{\sf T}\boldsymbol{\eta}_1\leq 0
\end{align}
Further, observe that there exists a small constant $\rho>0$ satisfying
\begin{align}
\label{eq:LMI1-dash}
&\begin{bmatrix}
    (\mathcal{A}_{i}+\mathcal{B}'K)^{\sf T}P+P(\mathcal{A}_i+\mathcal{B}'K)+\rho I& 
    \mathcal{K}'^{\sf T}\color{black}-P
    \begin{bmatrix}\mathcal{B}' & \mathcal{L}'\end{bmatrix}\\
    \mathcal{K}'-
    \begin{bmatrix}
    \mathcal{B}'^{\sf T}\\
    \mathcal{L}'^{\sf T}
    \end{bmatrix}P& -2I
\end{bmatrix}
\preceq 0, \ \ \ i=2,\ldots,N,
\end{align}
\begin{align}
    \label{eq:LMI2-dash}
&\begin{bmatrix}
 (\mathcal{A}_1+\mathcal{B}_1\check{K})^{\sf T}\check{P}+\check{P}(\mathcal{A}_1+\mathcal{B}_1\check{K})+\rho I& \check{\mathcal{K}}^{\sf T}-\check{P}\begin{bmatrix}
     \mathcal{B}_1 & \mathcal{L}_1
 \end{bmatrix}\\
 \check{\mathcal{K}}-
 \begin{bmatrix}
     \mathcal{B}_1^{\sf T}\\
     \mathcal{L}_1^{\sf T}
 \end{bmatrix}
 \check{P} & -2I
\end{bmatrix}
\preceq 0,
\end{align}
if the strict matrix inequalities \eqref{eq:LMI1}-\eqref{eq:LMI2} have their solutions $P$ and $\check{P}$. 
By combining \eqref{eq:Vdot}-\eqref{eq:LMI2-dash}, we get
\begin{eqnarray}
\label{eq:Vdot-2}
    \dot{V}(\xi)+\rho \|\xi\|^2\color{black}\leq 2\boldsymbol{\eta}^{\sf T}(\boldsymbol{\eta}-\mathcal{K}\color{black}\xi)+2\xi^{\sf T}\mathcal{P}\mathcal{B}\boldsymbol{w}^{\star}
\end{eqnarray}
with $\mathcal{K}:=\mathrm{diag}\{\check{\mathcal{K}},I_{N-1}\otimes \mathcal{K}'\}$. From \eqref{eq:incr-sect-bound-input}, \eqref{eq:Psi-sect}, \eqref{eq:w-def-incr}, \eqref{eq:g-def-incr}, and orthogonality of $U$, the following incremental sector bound conditions also hold true. 
\begin{align}
\label{eq:bold-w-cond}
\boldsymbol{w}^{\sf T}(\boldsymbol{w}-\gamma\boldsymbol{u})&\leq 0\quad\forall \boldsymbol{u}\in\mathds{R}^{Nm},\\
\label{eq:bold-g-cond}
    \boldsymbol{g}^{\sf T}(\boldsymbol{g}-\mu'\boldsymbol{x})&\leq 0\quad\forall\boldsymbol{x}\in\mathds{R}^{Nn}.
\end{align}

Combining \eqref{eq:der1}-\eqref{eq:bold-g-cond} yields 
\begin{align*}
&2\boldsymbol{\eta}^{\sf T}(\boldsymbol{\eta}-\mathcal{K}\xi)\\
&\leq 2\boldsymbol{w}^{\sf T}(\boldsymbol{w}- \gamma\color{black}\mathrm{diag}\{\check{K},(I_{N-1}\otimes K)\}\xi)+2\boldsymbol{g}^{\sf T}(\boldsymbol{g}-\mu\{\mathrm{diag}\{\check{J},(I_{N-1}\otimes J)\}\xi)\\
&=2\boldsymbol{w}^{\sf T}(\boldsymbol{w}-\gamma\color{black}\boldsymbol{u})+2\boldsymbol{g}^{\sf T}(\boldsymbol{g}-\mu'\boldsymbol{x})\leq 0 \ \ \ \forall \xi,\boldsymbol{\eta}.
\end{align*}
Thus, 
\begin{align}
\dot{V}({\xi})&\leq -\rho \|\xi\|^2+2\|\mathcal{B}\|\max\{\overline{\lambda}(P), \overline{\lambda}(\check{P})\}\|\xi\|\|\boldsymbol{w}^\star\| \\
&\leq -(1-\theta)\rho\|\xi\|^2+\rho \theta\|\xi\|^2+2\|\mathcal{B}\|\max\{\overline{{\lambda}}(P), \overline{\lambda}(\check{P})\}\|\xi\|\|\boldsymbol{w}^\star\|
\ \ \ \ \forall t\geq 0 
\end{align}
is satisfied for some $\theta\in\left]0,1\right[$.

Next, we apply a well-known ISS criteria \citep{Khalil2002} to analyze the convergence. To this aim, define 
\begin{align*}
    \underline{\alpha}(\xi)&:=\min\{\underline{\lambda}(P),\underline{\lambda}(\check{P})\}\|\xi\|^2,\\
    \overline{\alpha}(\xi)&:=\max\{\overline{\lambda}(P),\overline{\lambda}(\check{P})\}\|\xi\|^2, 
      \\
        W(\xi)&:= (1-\theta)\rho\|\xi\|^2
,    \\
    r(\|\boldsymbol{w}^\star\|)&:=\frac{2\|\mathcal{B}\|\max\{\overline{\lambda}(P),\overline{\lambda}(\check{P})\}}{\theta \rho}\|\boldsymbol{w}^\star\|,
\end{align*}
and observe that 
\begin{subequations}
\label{eq:ISS-condition}
\begin{align}
\underline{\alpha}(\xi)&\leq V(\xi)\leq \overline{\alpha}(\xi),  \\
\dot{V}(\xi)&\leq-W(\xi) \ \ \ \ \ \forall \|\xi\|\geq r(\|\boldsymbol{w}^\star\|)
\end{align}
\end{subequations}
hold. Also, notice that $\underline{\alpha}(\xi), \overline{\alpha}(\xi)$ are the class-$\mathscr{K}_\infty$ functions, $r:\mathds{R}_+\rightarrow\mathds{R}_+$ is the class $\mathscr{K}$ function, and $W(\xi)$ is the positive definite function. Applying Theorem 4.19 in \citep{Khalil2002}, one can conclude that the closed-loop multi-agent system is ISS with respect to the external input $\boldsymbol{w}^\star(t)$ to the state $\xi(t)$. That is,
\begin{align*}
    \|\xi(t)\|\leq \kappa_1(\|\xi(0)\|,t)+\kappa_2(\sup_{0\leq \tau \leq t}\|\boldsymbol{w}^\star(\tau)\|), \ \ t\geq0
\end{align*}
    holds for some class-$\mathscr{KL}$ function $\kappa_1 $ and class-$\mathscr{K}$ function $\kappa_2$. Namely, this means that the closed-loop multi-agent system \eqref{eq:collective-dynamics1} is ISS from the input $\boldsymbol{w}^\star(t)$ to the state $\xi(t)$. To be more precise, the function $\kappa_2$ is given by 
    \begin{align}
    \kappa_2(\sup_{0\leq\tau\leq t}\|\boldsymbol{w}^\star(\tau)\|)
    =2\sqrt{\frac{\max\{\overline{\lambda}(P),\overline{\lambda}(\check{P})\}}{\min\{\underline{\lambda}(P),\underline{\lambda}(\check{P})\}}}\frac{\|\mathcal{B}\|\max\{\overline{\lambda}(P),\overline{\lambda}(\check{P})\}}{\theta \rho}\sup_{0\leq\tau\leq t}\|\boldsymbol{w}^\star(\tau)\|
    \end{align}
By recalling the sector condition in \eqref{eq:bold-w-cond}, one can observe that
\[
\sup_{0\leq \tau\leq t}\|\boldsymbol{w}^\star(\tau)\|\leq\gamma \|\overline{u}\|.
\] 
Thus, 
\begin{align}
\limsup_{t\rightarrow\infty} \|\xi(t)\|\leq 
2\sqrt{\frac{\max\{\overline{\lambda}(P),\overline{\lambda}(\check{P})\}}{\min\{\underline{\lambda}(P),\underline{\lambda}(\check{P})\}}}\frac{\|\mathcal{B}\|\max\{\overline{\lambda}(P),\overline{\lambda}(\check{P})\}}{ \rho}\|\overline{u}\|,
\end{align}
where we used $\theta\in\left]0,1\right[$. From the facts 
\[
\|\xi(t)\|\geq\|x(t)-\overline{x}(t)\|=\|x(t)-\mathds{1}_N\otimes x^\star\| \ \ \ \forall t \geq 0
\] and $\|U\|=\|U^{\sf T}\|=1$ due to orthogonality of $U$, we thus conclude:  
\begin{align}
\label{eq:x-xstar}
\limsup_{t\rightarrow\infty}\|x(t)-\mathds{1}_N\otimes x^{\star}\|\leq 2\sqrt{\frac{\max\{\overline{\lambda}(P),\overline{\lambda}(\check{P})\}}{\min\{\underline{\lambda}(P),\underline{\lambda}(\check{P})\}}}\frac{\|\mathcal{B}\|\max\{\overline{\lambda}(P),\overline{\lambda}(\check{P})\}}{\rho}\|\overline{u}\|.
\end{align}

Summarizing the above discussion, the bounded state agreement to the optimizer $x^\star$ is achieved, where the sub-optimality certificate exists as in the RHS in \eqref{eq:x-xstar}. This means that the robust distributed optimal coordination problem (Problem~\ref{prob:1}) is solved by the proposed dynamical protocol. This concludes the proof.
\end{proof}
\begin{remark}
One can verify that stabilizability of the pair $(A,B_0)$ guarantees the existence of the feedback gains $K$ and $\check{K}$ stabilizing the \emph{augmented} systems. Also, Hurwitz stability of $\mathcal{A}_1+\mathcal{B}_1\check{K}$, $\mathcal{A}_i+\mathcal{B}_i K$ $(i=2,\ldots,N)$ implies that the constant vectors $(\overline{x}, \overline{v}, \overline{\zeta}, \overline{\eta})$ satisfying the algebraic equations \eqref{eq:ss1}-\eqref{eq:ss4} uniquely exist due to the construction of \eqref{eq:collective-dynamics1}. 
\end{remark}
\begin{remark}
The LHSs of \eqref{eq:LMI1} and \eqref{eq:LMI2} are affine in $\lambda_i$ respectively. It thus suffices to check only at the extreme points $\lambda_2$ and $\lambda_N$, and the LMIs are free from the size of the network.
\end{remark}
\begin{remark}
If the gain $K$ is handled as a decision variable, the matrix inequalities \eqref{eq:LMI1}-\eqref{eq:LMI2} have some products of the decision variables $P$, $\check{P}$ and $K$. Some of BMI solvers are applicable e.g. \citet{BMI1}. Otherwise, one can check the LMI condition for a gain $K$, which is chosen according to trial and error. A similar trial and error procedure is sometimes required in synthesis of optimizing controller in the literature, such as \citep{Romano2}. It should also be noted that the initial guess can be easily obtained by applying a well-known change of variables technique to the matrix inequality in \eqref{eq:LMI1} (refer to e.g. \citet{boyd1994linear}). 
\end{remark}
\begin{remark}
The actual value of the sub-optimality bound (RHS in \eqref{eq:x-xstar}) depends on the input $\overline{u}$, and thus it is difficult to calculate a priori. Nevertheless, the existence of this a posteriori bound can guarantee the solvability of the considered problem by the proposed control protocol, even in the presence of the time-varying input nonlinearities. 
\end{remark}
\section{Numerical Simulation}
Let us consider five agents whose dynamics is given by 
\[
\dot{x}_i=
\begin{bmatrix}
0.2 & 0.6\\
-0.6 & 0
\end{bmatrix}x_i+\begin{bmatrix}
1 & 5\\
2 & 3
\end{bmatrix}\phi_i(u_i), \ i=1,\ldots, 5.
\] 
Inspired by \citep{Takaba-ICCAS}, we consider the input nonlinearity given in the form
\begin{align*}
[\phi_i(u_i,t)]_j:=0.8[u_i]_j+0.2\sin(2t)[u_i]_j,~j=1,2.
\end{align*}
This nonlinearity satisfies the sector condition with $\alpha=0.6, \beta=1$, and we set $\gamma=0.5$, which is a conservative estimation of the uncertainties. The communication is executed along the graph $\mathcal{G}$ depicted in Fig.~\ref{fig:graph-str}. 

\begin{figure}[pos=h]
    \centering
    \includegraphics[width=0.2\linewidth]{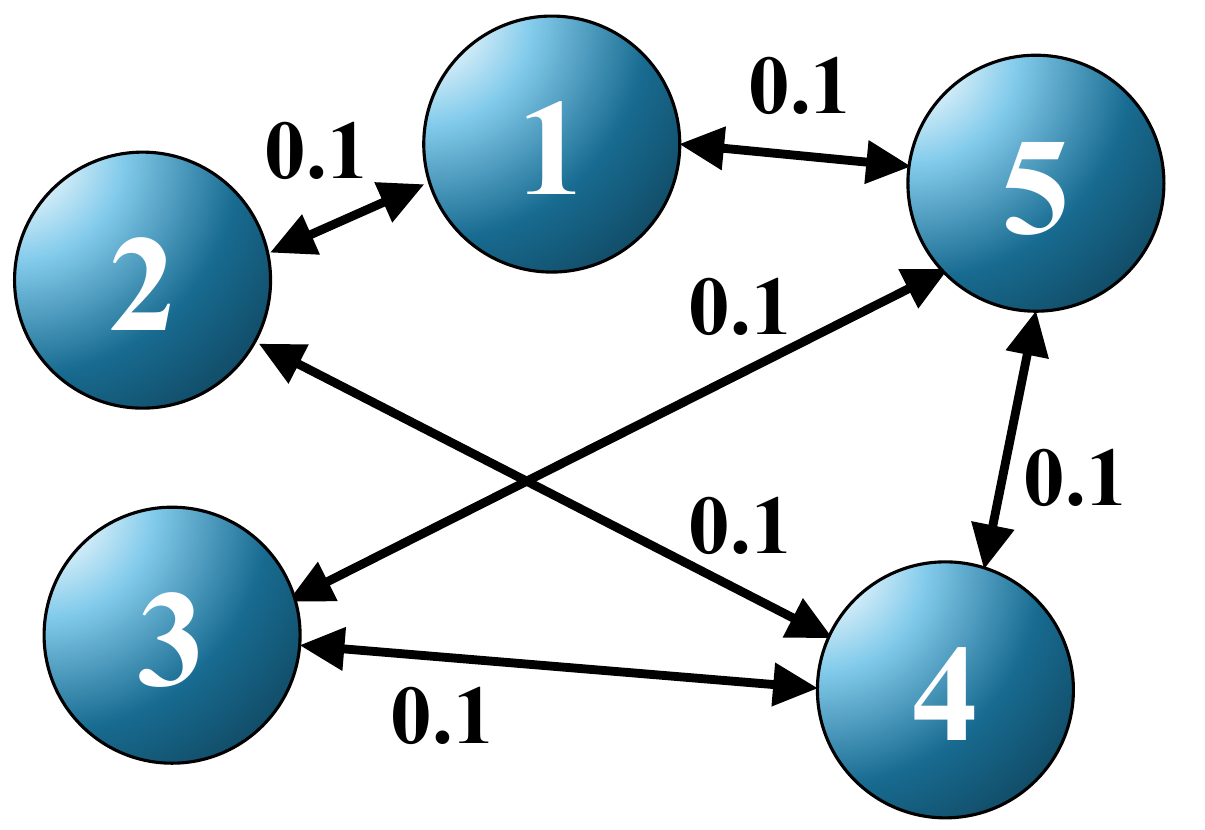}
    \caption{The topology of the graph $\mathcal{G}$}
    \label{fig:graph-str}
\end{figure}

The local objective functions are defined by 
\begin{align*}
f_1(x) &= \frac{1.1}{2}\left[(x_1+1)^2 + x_2^2\right], \
f_2(x) = \frac{1.1}{2}\left[(x_1+0.75)^2 + (x_2-0.25)^2\right], \\ 
f_3(x) &= \frac{1.1}{2}\left[(x_1+0.5)^2 + (x_2-0.5)^2\right],\
f_4(x) = \frac{1.1}{2}\left[(x_1+0.25)^2 + (x_2-0.75)^2\right],\\ 
f_5(x) &= \frac{1.1}{2}\left[x_1^2 + (x_2-1)^2\right].
\end{align*}
We thus set $\mu=1,\ell=1.1$, and $\rho=0.1$. By defining the objective functions as above, the optimizer of the resulting optimization problem $(\mathbf{OP})$ is given by $x^\star=[-0.5,0.5]^{\sf T}$. As an initial guess of the gain $K$, by applying a well-known change of variable techniques to matrix inequality \eqref{eq:LMI1}, we obtain 
\[
K=\left[
\begin{array}{cc:cc:cc:cc}
1.8386 & -4.9411 & -0.4984 & 2.5780 & -0.1485 & 2.6716 & 0.0071 & -1.2094 \\
-2.1966 & 0.9602 & -0.4984 & -0.5028 & 1.0652 & -0.6965 & -0.5092 & 0.3012
\end{array}
\right].
\]
This gain also satisfies \eqref{eq:LMI2} simultaneously, and thus this gain meets the sufficient condition in Theorem~\ref{thm:1}. 

    We carry out the numerical simulations and computations using \textsf{MATLAB R2024b} with \textsf{YALMIP} \cite{YALMIP} and \textsf{MOSEK} \cite{MOSEK} on a MacBook Pro (Apple M3, 16 GB memory).

    Fig.~\ref{fig:state-traj} shows the trajectories of $x_i(t)$. We can observe that $x_i(t)$ approaches a neighborhood of the optimizer $[-0.5,0.5]^{\sf T}$ despite the input nonlinearities. Also from Fig.~\ref{fig:norm-traj}, the error $\|x(t)-\mathds{1}\otimes x^\star\|$ remains within a neighborhood of zero, while it oscillates because of the sinusoidal time-varying input nonlineraties. 
    In this setup, the theoretical sub-optimality bound is given by $\epsilon\approx 2.7\times 10^{4}$, which is conservative, while the actual performance is much better than the theoretical guarantee. This conservativeness arises because the derived solvability condition in Theorem~\ref{thm:1} is a sufficient condition based on the sector conditions, and requires simultaneous satisfaction of the matrix inequalities. Similar drawbacks have also been mentioned in the previous literature on robust (bounded) agreement, such as \cite{Lewis}. Developing less conservative methods remains future work.

\begin{figure}[pos=h]
    \centering
\includegraphics[width=0.6\linewidth]{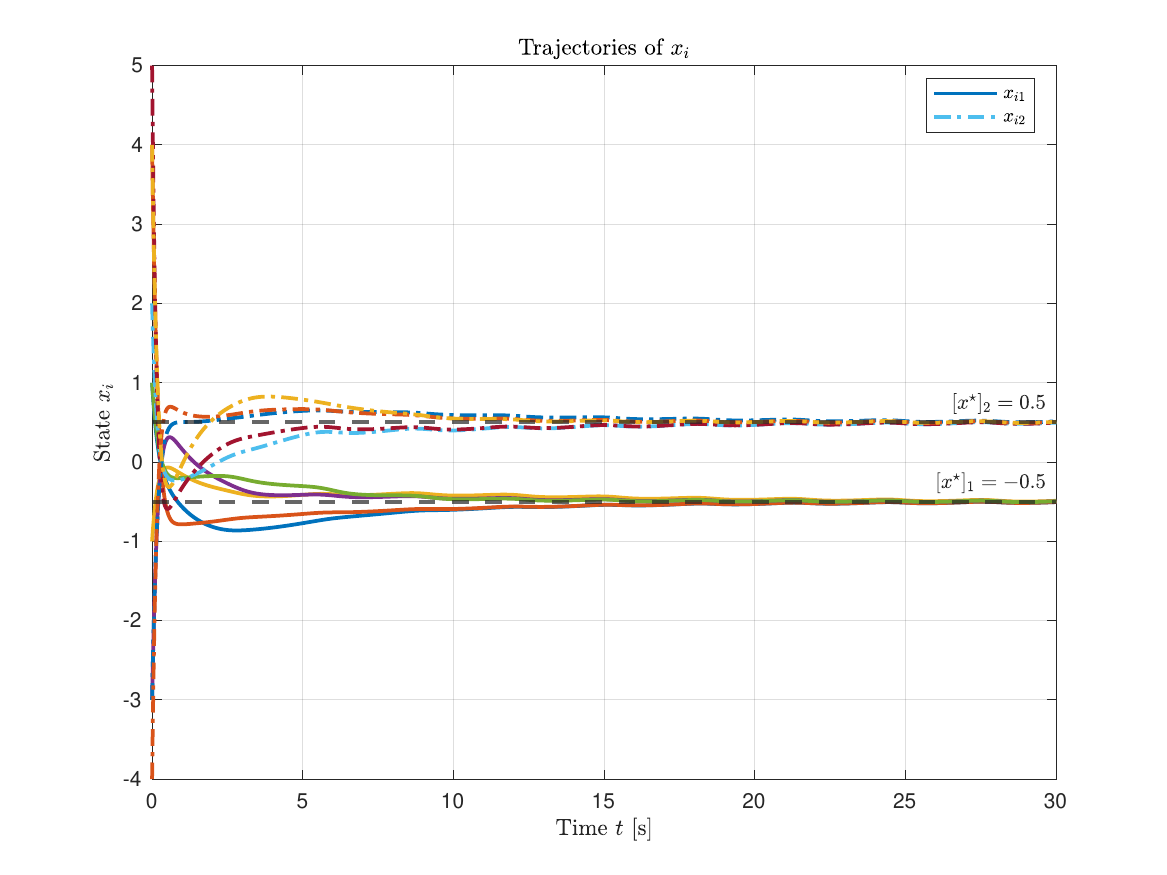}
    \caption{Trajectories of $x_i(t)$}
    \label{fig:state-traj}
\end{figure}

\begin{figure}[pos=h]
  \begin{subfigure}{0.48\textwidth}
    \centering
\includegraphics[width=0.95\linewidth]{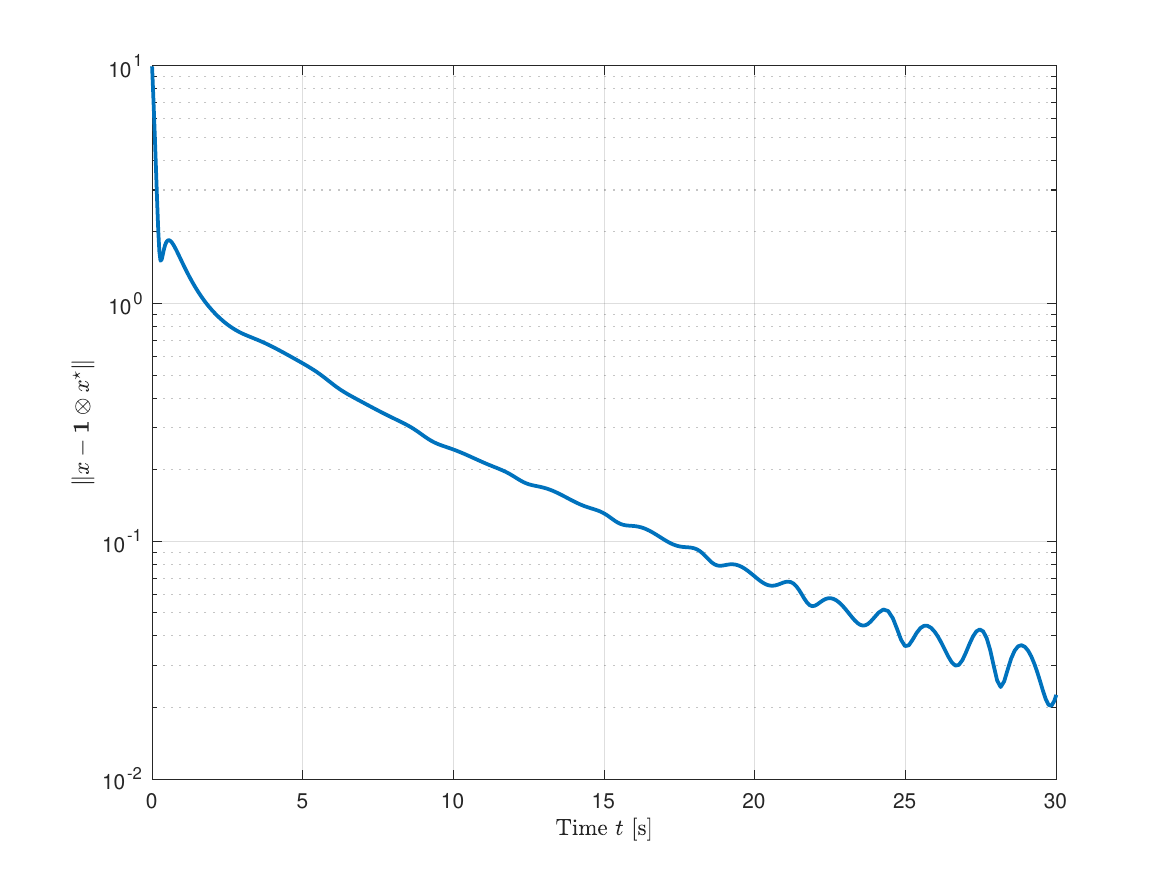}
    \caption{Trajectories of $\|x-\mathds{1}_N\otimes x^\star\|$}
    \label{fig:norm-traj}
\end{subfigure}
\begin{subfigure}{0.48\textwidth}
\includegraphics[width=0.95\linewidth]{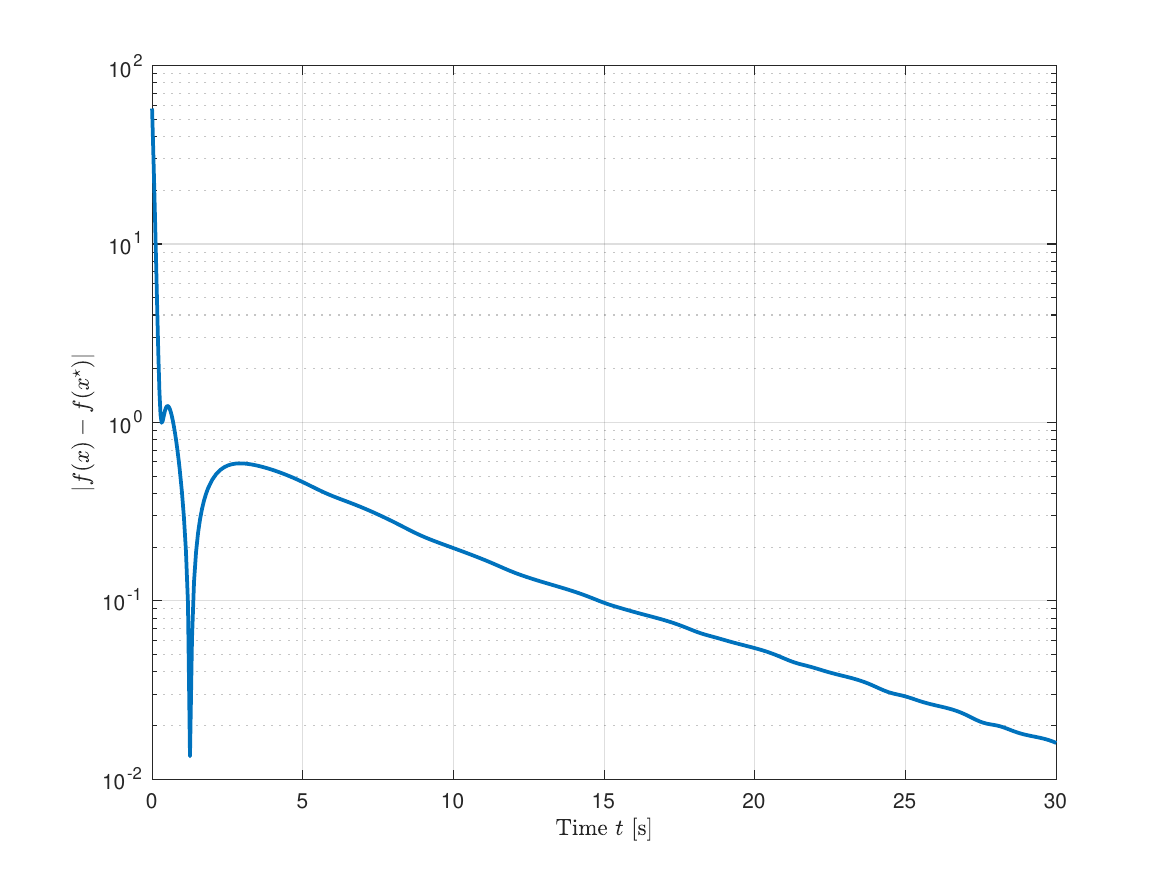}
    \caption{Trajectories of $|f(x)-f(x^\star)|$}
    \label{fig:obj-traj}
    \end{subfigure}
\end{figure}
\section{Concluding Remarks}
In this paper, we have studied robust distributed optimal coordination of linear agents with input nonlinearities. We have captured the time-varying and uncertain input nonlinearities as the sector-bounded uncertainties. By invoking the input-to-state stability analysis in robust control theory, we have developed a robust protocol design method for robust distributed optimal coordination, and have derived a novel solvability condition characterized in terms of the matrix inequalities.

\section*{Acknowledgment}
The author would like to thank Prof.~ Kiyotsugu Takaba (Ritsumeikan Univ.) for his valuable comments in our dairy discussions. Also, the author is sincerely grateful to Dr. Saeed Ahmed (Univ.~of Groningen) for his meaningful comments early on.

\bibliographystyle{model1-num-names}
\bibliography{260402_DOC_Bib}
\end{document}